\documentclass[review]{elsarticle}

\usepackage[margin=2.5cm]{geometry}

\usepackage{amssymb}

\usepackage{amsmath,amsthm}
\usepackage{amssymb, amsbsy}
\usepackage{gensymb}

\usepackage{graphicx}

\usepackage{import}
\usepackage{enumitem, outlines}
\usepackage{lineno}

\usepackage[edges]{forest}

\usepackage{setspace}
\usepackage[list=true]{subcaption}
\usepackage{dcolumn, multirow}
\usepackage{array}

\usepackage{doi}

\newtheorem{theorem}{Theorem}
\newtheorem{lemma}{Lemma}
\newtheorem{corollary}{Corollary}
\newtheorem{definition}{Definition}
\newtheorem{observation}{Observation}

 \newtheorem{case}{Case}
 \newtheorem{subcase}{Subcase}

\journal{}

\begin{document}

\begin{frontmatter}



\title{Parking Problem by Oblivious Mobile Robots in Infinite Grids}


\author[inst1]{Abhinav Chakraborty\corref{mycorrespondingauthor}}
\cortext[mycorrespondingauthor]{Corresponding author}
\ead{abhinav.chakraborty06@gmail.com}

\author[inst2]{Krishnendu Mukhopadhyaya}
\ead{krishnendu.mukhopadhyaya@gmail.com }

\affiliation[inst1]{organization={Department of Computer Science and Engineering, Institute of Technical Education and Research, Siksha O Anusandhan University, Bhubaneswar, Odisha, India}}

\affiliation[inst2]{organization={Advanced Computing and Microelectronics Unit, Indian Statistical Institute},
            addressline={203, B. T. Road}, 
            city={Kolkata},
            postcode={700108}, 
            state={West Bengal},
            country={India}}

\begin{abstract}
In this paper, the parking problem of a swarm of mobile robots has been studied. The robots are deployed at the nodes of an infinite grid, which has a subset of prefixed nodes marked as \textit{parking nodes}. Each parking node $p_i$ has a capacity of $k_i$ which is given as input and equals the maximum number of robots a parking node can accommodate. As a solution to the parking problem, robots need to partition themselves into groups so that each parking node contains a number of robots that are equal to the capacity of the node in the final configuration. It is assumed that the number of robots in the initial configuration represents the sum of the capacities of the parking nodes. The robots are assumed to be autonomous, anonymous, homogeneous, identical and oblivious. They operate under an asynchronous scheduler. They neither have any agreement on the coordinate axes nor do they agree on a common chirality. All the initial configurations for which the problem is unsolvable have been identified. A deterministic distributed algorithm has been proposed for the remaining configurations, ensuring the solvability of the problem.
\end{abstract}




\end{frontmatter}



\section{Introduction}

Robot swarms are groups of generic mobile robots that can collaboratively execute complex tasks. Such systems of mobile robots are assumed to be simple and inexpensive and offer several advantages over traditional single-robot systems, such as scalability, robustness and versatility. A series of research on the algorithmic aspects of distributed coordination of robot swarms has been reported in the field of distributed computing (see \cite{flocchini2019distributed} for a comprehensive survey). In the traditional framework of swarm robotics, the robots are assumed to be anonymous (no unique identifiers), autonomous (there is no centralized control), identical (no unique identifiers), homogeneous (each robot executes the same deterministic distributed algorithm) and oblivious (no memory of past information) computational entities. The robots are represented as points in the Euclidean plane. They do not have access to any global coordinate system. However, each robot has its own local coordinate system, with the origin representing the current position of the robot. The robots do not have an explicit means of communication, i.e., they are assumed to be silent. They are disoriented, i.e., they neither agree on a common coordinate axes nor do they have any agreement on chirality. Each robot is equipped with visibility sensors, by which they can perceive the deployment region.  

\noindent In this paper, the deployment region of the robots is assumed to be an infinite square grid, which represents a natural discretization of the plane. The robots are deployed at the nodes of the input grid graph. The graph also consists of some prefixed grid nodes, designated as parking nodes. When a robot becomes active, it operates according to the \textit{Look-Compute-Move} cycle. A robot takes a snapshot of the entire graph, including the positions of the other robots and parking nodes in the \textit{Look} phase. Based on the snapshot, it computes a destination node in the \textit{Compute} phase according to a deterministic algorithm, where the destination node might be its current position as well. Finally, it moves towards the destination in the \textit{Move} phase. There are three types of schedulers considered in the literature which describe the timing of operations of the robots. In the fully synchronous model ($\mathcal{FSYNC}$), time is divided into global rounds, and all the robots are activated simultaneously. The robots take their snapshots simultaneously and then execute their moves concurrently. In the semi-synchronous model ($\mathcal{SSYNC}$), a subset of robots is activated at the same time. In this paper, we have considered the most general model, which is the asynchronous model ($\mathcal{ASYNC}$). In this setting, there is no common notion of time, and all the robots are activated independently. Each of the Look, Compute and Move phases has a finite but unpredictable duration. In the initial configuration, it has been assumed that the robots are placed at the distinct nodes of the grid graph. During the look phase, the robots can perceive the parking nodes using their visibility sensors. Each parking node has a capacity, which is subjected to a constraint that it can accommodate a maximum number of robots equal to its capacity. The capacity of a parking node is given as an input to each robot. For simplicity, we have assumed that the number of robots in the initial configuration is equal to the sum of the capacities of the parking nodes. In this paper, we have assumed that the robots have \textit{global-strong multiplicity detection capability}. This means the robots are able to determine the exact number of robots that make up the multiplicity in each node. It has been proved later that the parking problem is unsolvable if the robots do not have such capabilities. 

In the initial configuration, the robots are placed at the distinct nodes of the grid graph. During the look phase, the robots can perceive the parking nodes using their visibility sensors. The parking nodes occupy distinct nodes of the grid. We assume that a robot may be deployed on a parking node in the initial configuration. In the initial configuration, each parking node has a capacity, which is subjected to a constraint that, at any moment of time, it can accommodate a maximum number of robots equal to its capacity. The capacity of a parking node is given as an input to each robot. For simplicity, assume that each of the parking nodes has uniform capacities, i.e., the capacities of the parking nodes are equal. Further, assume that the number of robots in the initial configuration equals $pk$, where $p$ is the number of parking nodes.

\subsection{Motivation}
 The fundamental motivation behind studying the parking problem is twofold.
 The fundamental motivation behind studying the parking problem is twofold. Firstly, the parking problem can be viewed as a special case of the partitioning problem \cite{EFRIMA20091355}, which requires the robots to divide themselves into $m$ groups, each consisting of $k$ robots while converging into a small area. Unlike the partitioning problem, the parking problem requires that each parking node must contain robots exactly equal to its given capacity in the final configuration. However, the capacities of the parking nodes may be different. Moreover, if the capacities of each of the parking nodes are assumed to be $k$, i.e., they are equal in the initial configuration; the problem is reduced to the $k$-epf problem \cite{a14020062}, which is a generalized version of the embedded pattern formation problem, where each fixed point contains exactly $k$ robots in the final configuration. 
 \noindent Secondly, in the traditional models, the robots are assumed to be points that can move freely on the plane. The robots are assumed to move with high accuracy and by infinitesimal distance in the continuous domain. Even if the area of robot deployment is small, a dimensionless robot can move without causing any collision. In practice, it may not always be possible to perform such infinitesimal movements with infinite precision. However, in our paper, the robots are deployed at the nodes of an infinite grid. The movements of the robots are restricted along the grid lines, and a robot can move toward one of its neighbors at any instant of time. The restrictions imposed by the grid model on the movements of the robots make it challenging to design collision-less algorithms, as opposed to the movement of the robots in a continuous environment. In addition to the theoretical benefits, the parking nodes can also be seen as base stations or charging stations with some allowable capacities.

 \subsection{Related Works}

Most of the theoretical studies on swarm robotics have been concentrated on \textit{arbitrary formation problem} and \textit{gathering} under different settings. The Arbitrary Pattern Formation or $\mathcal APF$ is a fundamental coordination problem in Swarm Robotics, where the robots are required to form any specific but arbitrary geometric pattern given as input. The study of $\mathcal APF$ was initiated in \cite{DBLP:journals/jfr/SugiharaS96}. The authors characterized the class of formable patterns by using the notion of symmetricity, which is essentially the order of the cyclic group that acts on the initial configuration. The $\mathcal APF$ was first studied in the $\mathcal{ASYNC}$ by Flocchini et al. \cite{DBLP:journals/tcs/FlocchiniPSW08}, where the robots are assumed to be oblivious. While all the previous studies considered the problem with unlimited visibility, Yamauchi et al. \cite{10.1007/978-3-319-03578-9_17} studied the problem where the robots have limited visibility. Cicerone et al. \cite{DBLP:journals/dc/CiceroneSN19} studied the $\mathcal APF$ problem without assuming common chirality among the robots. Bose et al. \cite{BOSE2020213} were the first to study the problem in a grid-based terrain, where the movements of the robots are restricted only along grid lines and only by a unit distance at each step. Bose et al. \cite{BOSE2021138} investigated the problem on the Euclidean plane in a setting where the robots are assumed to be opaque, i.e., the view of a robot can be obstructed by the presence of the other robots. The $\mathcal APF$ problem was studied by Cicerone et al. \cite{10.1145/3427796.3427833} in regular tessellation graphs, where the initial configuration is assumed to be asymmetric. The gathering problem may be viewed as a point formation problem. The problem requires the robot to gather at one of the points not known beforehand and within a finite amount of time. It has been extensively studied in the literature \cite{10.1007/3-540-45061-0_90,DBLP:conf/soda/AgmonP04,DBLP:journals/jda/BhagatCM16}. 

\noindent D'Angelo et al. \cite{DBLP:journals/tcs/DAngeloSKN16}, studied the gathering problem on finite grids. Stefano et al. studied the optimal gathering problem in infinite grids \cite{DBLP:journals/iandc/StefanoN17}. In this paper, they proposed an optimal deterministic algorithm that minimizes the total distance traveled by all the robots. The concept of fixed points was first introduced by Fujinaga et al. \cite{DBLP:conf/opodis/FujinagaOKY10} on the Euclidean plane. In this paper, the \textit{landmark} covering problem was studied. The problem requires that each robot must attain a configuration where all the robots must occupy a single fixed point or landmark. They propose an algorithm based on the assumption that the robots agree on a common chirality. The proposed algorithm minimizes the total distance traveled by all the robots. In \cite{DBLP:journals/dc/CiceroneSN19a}, Cicerone et al. studied the \textit{embedded pattern formation} problem without assuming any common chirality among the robots. The problem necessitates a distributed algorithm in which each robot must occupy a unique fixed point within a finite amount of time. The \textit{$k$-circle formation problem} \cite{a14020062,DBLP:journals/tcs/DasCBM22} has been studied in the setting where the robots agree on the directions and orientations of the $Y$- axis and on the disoriented setting. Given a positive integer $k$, the \textit{$k$-circle formation problem} asks a swarm of mobile robots to form disjoint circles. Each of these circles must be centered at one of the pre-fixed points on the plane. Each circle must contain a total of $k$ robots at distinct locations on the circumference of the circles. Bhagat et al. \cite{a14020062} also studied the \textit{$k$- epf problem} in the continuous domain, which is a generalized version of the embedded pattern formation problem. This problem necessitates the arrival and retention of exactly $k$ robots at each fixed point. Cicerone et al. \cite{DBLP:journals/dc/CiceroneSN18} studied a variant of the gathering problem, where each robot must gather at one of the prefixed meeting points. The problem was defined as \textit{gathering on meeting points} problem. The authors proposed a deterministic algorithm that minimizes the total distance traveled by all the robots and minimizes the maximum distance traveled by a single robot. \textit{Gathering over meeting nodes} problem was studied by Bhagat et al. \cite{DBLP:journals/fuin/BhagatCDM22,DBLP:journals/ijfcs/BhagatCDM23}. In this problem, the robots are deployed on the nodes of an infinite square grid, which has a subset of nodes marked as meeting nodes. Each robot must gather at one of the prefixed meeting nodes within a finite amount of time. Other specific problems that have been considered in the infinite grid domain are the \textit{Mutual visibility problem} \cite{DBLP:journals/ijnc/SharmaVT21,DBLP:conf/ipps/HectorVST20,ADHIKARY2022}, \textit{exploration problem} \cite{10.1007/978-3-030-67087-0_9}, etc. The Mutual Visibility problem requires a distributed algorithm that allows the robots to reposition themselves to form a configuration in which they occupy distinct locations, and no three of them are collinear. The objective of the exploration problem is to visit every node of the graph.

\subsection{Our Contribution}
This paper considers the parking problem over an infinite grid. The robots are deployed at the nodes of an infinite grid, which also consists of some prefixed parking nodes. Each parking node $p_i$ has a capacity $k_i$, which is the maximum number of robots it can accommodate at any moment of time. We assume that the number of robots $n$ is equal to $\sum\limits_{i=1}^m k_i$, where $m$ is the total number of parking nodes. The robots are assumed to be anonymous, autonomous, homogeneous and oblivious. The robots are activated under a fair asynchronous scheduler. Under this setup, we have characterized all the initial configurations and the values of $k_i$ for which the problem is unsolvable. For the remaining configurations, a deterministic algorithm was proposed. 
\subsection{Organisation}
In the next section the robot model and the definitions relevant to this paper has been described. Some basic definitions and notations are presented in this section. In section \ref{s3}, the configurations for which the parking problem is unsolvable have been identified. In this section, a partitioning of the initial configurations has been specified. In section \ref{s4}, a formal description of the algorithm has been mentioned. Section \ref{chap6:sec6} discusses the correctness proofs of the algorithm. The paper is concluded with some potential future directions in Section \ref{chap6:sec7}. 
\section{Models and Definitions}
\subsection{Models}

The robots are assumed to be dimensionless, anonymous, autonomous, identical, homogeneous and oblivious. The robots are assumed to be disoriented, i.e., they neither have any agreement on the coordinate axes nor have any agreement on a common chirality. They do not have an explicit means of communication, i.e., they are assumed to be silent. Let $P=(\mathbb{Z}$, $E')$ denote the infinite path graph with the vertex set $V$ corresponding to the set of integers $\mathbb {Z}$ and the edge set is denoted by the ordered pair $E'=\lbrace (i$, $i+1)| i\in \mathbb{Z}\rbrace$. Let $\mathcal R=\lbrace r_1, r_2 \ldots r_n\rbrace$ denote the set of robots that are deployed at the nodes of $G$, where $G$ is the input infinite grid graph defined as the usual \textit{Cartesian Product} of the graph $P \times P$. Let $r_i(t)$ denote the node occupied by the robot $r_i\in \mathcal R$ at time $t$. Assume that $\mathcal R(t)$ denotes the set of all such distinct nodes occupied by the robots in $\mathcal R$ at time $t$. Since the robots are deployed at the nodes of an infinite square grid, they have an agreement on a common measure of unit distance. The input grid graph also comprises some prefixed nodes designated as \textit{parking nodes}. Let $\mathcal P= \lbrace p_1,p_2,..., p_m \rbrace$ denote the set of parking positions. In the initial configuration, the parking nodes are located at the distinct nodes of the grid. A robot may be deployed at one of the parking nodes in the initial configuration.

\noindent The movements of the robots are restricted along the grid lines. At any instant of time, a robot can move only to one of its four neighboring nodes. The movement of the robot is assumed to be instantaneous, i.e., the robot can be observed only at the nodes of the graph and not on the edges. In other words, no robot can be seen while moving. A robot's vision is assumed to be global, meaning that each robot is equipped with visibility sensors that allow it to observe the whole grid graph. 

\subsection{Terminologies and Definitions}
The remainder of the paper uses the notations and terminologies listed below.
\begin{itemize}
    \item \textbf{Distance between two nodes:} Let $d(u,v)$ denote the distance between two nodes $u$ and $v$. 
    \item \textbf{Capacity of a parking node:} The \textit{capacity} of a parking node given as an input is defined as the maximum number of robots the parking node can accommodate. A parking node is said to be \textit{saturated} if it contains exactly the number of robots equal to its capacity. A parking node is said to be \textit{unsaturated} if it is not saturated. Let $\mu:V\rightarrow \mathbb{N} \cup \lbrace 0 \rbrace$ be defined as a function, where:
    	\[ \mu(v)=\begin{cases} 
		0 & \text{if} \;v \textrm { is not a parking node} \\
		\textit{\textrm{ capacity of the parking node}} & \textrm{ otherwise} \\
		
		\end{cases}
		\]
     In the initial configuration, let $k_i$ be the capacity of a parking node $p_i$, $\forall i = 1, 2, \ldots, m$. 
    
    \item \textbf{Symmetry of a configuration $\boldsymbol{C(t)}$:} Two graphs $G_1= (V_{G_1}, E_{G_1})$ and $G_2= (V_{G_2}, E_{G_2})$ are said to be isomorphic if there exists a bijection $\phi: V_{G_1} \rightarrow V_{G_2}$ such that any two nodes $u,v \in V_{G_1}$ are adjacent in $G_1$ if and only if $\phi(u)$, $\phi(v)$ $ \in V_{G_2}$ are adjacent in $G_2$. An automorphism on a graph $G$ is a permutation of its nodes mapping edges to edges and non-edges to non-edges. Let $\lambda_t$ be defined as a function that denotes the number of robots residing on $v$ at time $t$. Without any ambiguity, we denote the function $\lambda_t$ by $\lambda$. $C(t)=(\mathcal R(t)$, $\mathcal P, \lambda, \mu$) denotes the \textit{system configuration} at any time $t$. An automorphism of a graph can be extended to the automorphism of a configuration. Two configurations are said to be isomorphic if there exists an \textit{automorphism} $\phi$ of the input grid graph such that $\lambda(v)=\lambda(\phi(v))$ and $\mu(v)=\mu(\phi(v))$, for all $v\in V$. The set of all automorphisms of a configuration forms a group which is denoted by $Aut(C(t)$, $\lambda, \mu)$. If $\vert Aut(C(t), $ $\lambda, \mu)\vert=1$, then the configuration is asymmetric. Otherwise, the configuration is said to be symmetric. We assume that the infinite grid is embedded in the \textit{Cartesian plane}. As a result, a grid can admit only three types of automorphism and combinations of them, 
    \begin{enumerate}
        \item \textit{translation}: defined by the shifting of the nodes to the same extent.
        \item \textit{reflection}: defined by the line of reflection axes.
        \item \textit{rotation}: defined by the angle of rotation and the center of rotation.
    \end{enumerate} 
    The reflection axis can be horizontal, vertical or diagonal. It can either pass through the nodes or edges of the grid. If a configuration admits rotational symmetry, then the center of rotation can be either a node, the center of an edge or the center of the area surrounded by four nodes. The angle of rotation can be either $90^{\circ}$ or $180^{\circ}$. Since the number of occupied nodes is finite, a translation symmetry is not admissible. Let $MER$ be the \textit{minimum enclosing grid} containing all the occupied nodes of $C(t)$. Assume that the dimension of $MER$ is $a\times b$. The number of grid edges on a side of $MER$ is used to define its length. 

    \begin{figure}[h]
				\centering
			{
				\includegraphics[width=0.27\columnwidth]{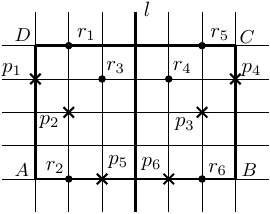}
			}
			\caption{The configuration is symmetric with respect to $l$. The crosses represent parking nodes and the black circles represent robot positions}
			\label{figure2}
		\end{figure}
  
\item \textbf{View:} Starting from one of the corners of $MER$, the entire grid is scanned in a direction parallel to the width of the rectangle. While scanning the grid, we associate the pair $(\lambda(v), \mu(v))$ to each node $v$ that the string encounters. Similarly, we can define the string associated with the same corner and encounter the nodes of the grid in the direction parallel to the length of the grid. Consider the eight senary strings of length $ab$ that are associated with the corners of $MER$, with two senary strings defined for each corner of $MER$. Let the two strings defined for a corner $i$ be denoted by $s_ {ij}$ and $s_ {ik}$.

\noindent If $MER$ is a non-square rectangle, we can distinguish between the two strings associated with a given corner by looking at the string that runs parallel to the side with the shortest length. Consider any particular corner $i$ of $MER$. Assume that $|ij|<|ik|$. We consider the direction parallel to $ij$ as the \textit{string direction} associated to $i$. We define $s_i=s_{ij}$ as the \textit{string representation} associated to the corner $i$. The direction parallel to the larger side is defined as the \textit{non-string direction} associated to the corner $i$. In the case of a square grid, between the two strings associated to a corner, the \textit{string representation} is defined as the larger lexicographic string, i.e., $s_i=max(s_{ij}, s_{ik})$, where the maximum is defined according to the lexicographic ordering of the strings. If the configuration is asymmetric, we will always get a unique largest lexicographic string. Without loss of generality, let $s_{i}$ be the largest lexicographic string among all the strings associated to the corners. Then we refer to $i$ as the {\it key corner}. If the configuration is asymmetric, the robots can be ordered according to the key corner and the string direction. A \textit{non-key corner} is defined as one that is not a key corner. In Figure \ref{figure2}, assume that the capacity of each parking node is 1. The lexicographic string associated with the corners $C$ and $D$ are $s_{CB}= s_{DA}$ ((0,0), (0,1), (0,0), (0,0), (0,0), (1,0), (0,0), (0,1), (0,0), (1,0), (0,0), (1,0), (0,0), (0,0), (0,1), (0,0), (0,0), (0,0), (0,0), (0,0), (0,0), (1,0), (0,0), (0,0), (0,1), (1,0), (0,0), (0,1), (0,0), (1,0), (0,0), (0,1), (0,0), (0,0), (0,0)). The strings $s_{CB}= s_{DA}=$ are the maximum lexicographic strings associated and hence $C$ and $D$ are the key corners. The \textit{configuration view} of a node is defined as the tuple ($d', x$), where $d'$ denotes the distance of a node from the key corner in the string direction and $x$ denotes the type of the node, i.e., $x$ is either an empty node, parking node or a robot position.
		\begin{figure}[ht]
				\centering
			{
				\includegraphics[width=0.28\columnwidth]{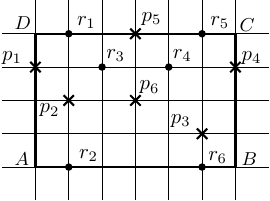}
			}
			\caption{Figure highlighting the definition of leading corner.}
			\label{figure1}
		\end{figure}

\item \textbf{Symmetricity of the set} $\mathbf{\mathcal P}$: We may define the symmetry of the set $\mathcal P$ in the same way as we define the symmetry of a configuration. The smallest grid-aligned rectangle that includes all the parking nodes is denoted as $M_{\mathcal P}$. 
		
\noindent We can define a string $\alpha_{i}$ similar to $s_{i}$. The only difference is that each node $v$ is associated with $\mu(v)$ instead of the pair $(\lambda(v), \mu(v))$. If the parking nodes are asymmetric, a unique lexicographic largest string $\alpha_i$ always exists. If the parking nodes are not asymmetric, then the parking nodes are said to be symmetric. The corner with which the lexicographic largest string $\alpha_i$ is associated is defined as the \textit{leading corner}. In Figure \ref{figure1}, assume that the capacity of $p_1=p_2=p_3=3$ and $p_4=p_5=p_6=2$. Then, $\alpha_{DA}=03000 00300 00000 20200 00000 00030 02000$ is the largest lexicographic string among the $\alpha_i's$ and hence we have $D$ as the leading corner. According to this definition of symmetricity of the set $\mathcal P$, the parking nodes that are located in the symmetric positions must have equal capacities.
\end{itemize}
\begin{definition}
Let $C(0)$ be any given initial configuration. A parking node $p_i$ is said to have a higher order than the parking node $p_j$ if it appears after $p_j$ in the string representation $\alpha_k$, associated to some leading corner $k$ of $MER$. Similarly, a robot $r_i$ has a higher order or has a higher configuration view than $r_j$ if it appears after $r_j$ in the string representation $s_k$, associated to some key corner $k$ of $MER$.
\end{definition}
According to this definition, if the configuration has two leading corners, then there exist two parking nodes having the highest order.

\section{Problem Definition and Impossibility Results} \label{s3}

\subsection{Problem Definition}
Let $C(t)=((R(t)$, $P, f, \lambda$) denote the system configuration at any time $t$. Each parking node $p_i$ has a capacity $k_i$, which is the maximum number of robots it can accommodate at any instant of time. For each parking node $p_i$, the capacity $k_i$ is given as an input. The number of robots is assumed to be equal to $\sum\limits_{i=1}^m k_i$, where $m$ is the total number of parking nodes located at the nodes of an infinite grid. In an initial configuration, all the robots occupy distinct nodes of the grid. The goal of the parking problem is to transform any initial configuration at some time $t >0$ into a configuration such that:
\begin{itemize}
    \item each parking node $p_i$ is saturated, i.e., $p_i$ contains exactly $k_i$ robots on it.
    \item each robot is stationary.
	\item any robot taking a snapshot in the look phase at time $t$ will decide not to move.
\end{itemize}
Note that if each $k_i=1$, the problem is reduced to the embedded pattern formation problem.

\subsection{Partitioning of the initial configurations}
All the initial configurations can be partitioned into the following disjoint classes.

\begin{enumerate}
    \item $\mathcal I_1$: The parking nodes are asymmetric (In Figure \ref{figure1}, the parking nodes are asymmetric).
    	\begin{figure}[h]
			\centering
			{
				\includegraphics[width=0.239\columnwidth]{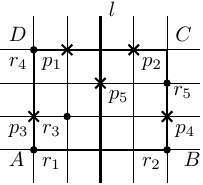}
			}
			\hspace*{0.21cm}
			{
				\includegraphics[width=0.20\columnwidth]{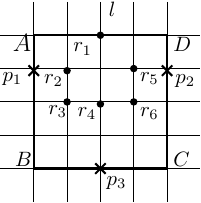}
			}
			\hspace*{0.21cm}
			{
				\includegraphics[width=0.26\columnwidth]{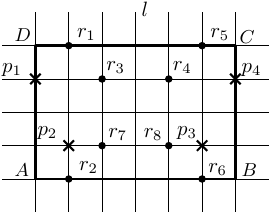}
			}
			\caption{Examples of $\mathcal I_{21}$, $\mathcal I_{221}$ and $\mathcal I_{222}$ configuration } 
			\label{sym1}
			\end{figure}
    \item $\mathcal I_2$: The parking nodes are symmetric with respect to a unique line of symmetry $l$. This class of configurations can be further partitioned into:
    \begin{enumerate}
			    \item $\mathcal I_{21}$: $C(t)$ is asymmetric (In Figure \ref{sym1}(a), with the assumption that each parking node has the same capacity 1, the configuration is asymmetric with the parking nodes symmetric with respect to $l$).
			    \item $\mathcal I_{22}$: $C(t)$ is symmetric with respect to $l$. This can be further partitioned into the following disjoint classes.
			   \begin{enumerate}
        \item $\mathcal I_{221}$: There exists at least one robot position on $l$ (In Figure \ref{sym1}(b), with the assumption that the capacity of each parking node is 2, $C(t)$ is symmetric with respect to $l$ and there exist robot positions $r_1$ and $r_4$ on $l$).
        \item $\mathcal I_{222}$: There does not exist any robot position on $l$. Also, there are no parking nodes on $l$ (In Figure \ref{sym1}(c), with the assumption that the capacity of each parking node is 2, $C(t)$ is symmetric with no robots or parking nodes on $l$).
        \item $\mathcal I_{223}$: There does not exist any robot position on $l$, but there exists at least one parking node on $l$ (In Figure \ref{sym5}(a), with the assumption that the capacity of each parking node not at $l$ is 1, $C(t)$ is symmetric with respect to $l$ and there exists parking node $p_5$ at $l$ with capacity 2).
    \end{enumerate}
    \end{enumerate}
    		\begin{figure}[h]
			\centering
			{
				\includegraphics[width=0.239\columnwidth]{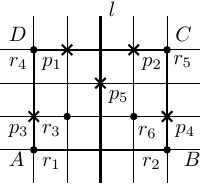}
			}
			\hspace*{0.21cm}
			{
				\includegraphics[width=0.200\columnwidth]{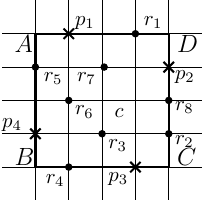}
			}
			\hspace*{0.21cm}
			{
				\includegraphics[width=0.205\columnwidth]{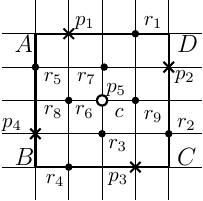}
			}
			\caption{Examples of $\mathcal I_{223}$ configuration, $\mathcal I_{31}$ configuration and $\mathcal I_{321}$ configuration. } 
			\label{sym5}
			\end{figure}
     \item $\mathcal I_3$: The parking nodes are symmetric with respect to rotational symmetry, with $c$ as the center of rotational symmetry. This class of configurations can be further partitioned into:
    \begin{enumerate}
        \item $\mathcal I_{31}$: $C(t)$ is asymmetric (In Figure \ref{sym5}(b), with the capacity of each parking node assumed to be 2, $C(t)$ is asymmetric with the parking nodes being symmetric with respect to rotational symmetry). 
        \item $\mathcal I_{32}$: $C(t)$ is symmetric with respect to $c$. This can be further partitioned into the following disjoint classes.
    \begin{enumerate}
        \item $\mathcal I_{321}$: There exists a robot position on $c$ (In Figure \ref{sym5}(c), with the assumption that the capacities of parking nodes $p_1$, $p_2$, $p_3$ and $p_4$ equal 1 and the capacity of the parking node $p_5$ equals 5, $C(t)$ is symmetric with respect to rotational symmetry. The robot $r_6$ is at the parking node $p_5$).
        \item $\mathcal I_{322}$: There does not exist a robot position or parking node on $c$ (In Figure \ref{sym2}(a), with the assumption that the capacity of each parking node is 1, $C(t)$ is symmetric with respect to rotational symmetry without any robot or parking node on $c$).
        \item $\mathcal I_{323}$: There exists a parking node on $c$, but no robot on $c$ (In Figure \ref{sym2}(b), with the assumption that the capacities of the parking nodes $p_1$, $p_2$, $p_3$ and $p_4$ equal 1 and the capacity of the parking node $p_5$ equals 4, $C(t)$ is symmetric with respect to rotational symmetry, with a parking node $p_5$ on $c$).
    \end{enumerate} 
    \end{enumerate}
     \end{enumerate}

    	\begin{figure}[h]
			\centering
			{
				\includegraphics[width=0.230\columnwidth]{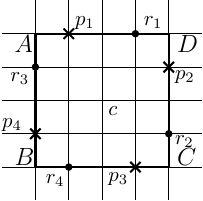}
			}
			\hspace{0.75cm}
				{
				\includegraphics[width=0.230\columnwidth]{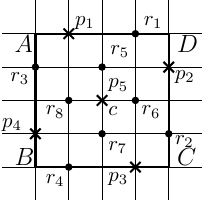}
			}
			\caption{Examples of $\mathcal I_{322}$ and $\mathcal I_{323}$ configuration.} 
			\label{sym2}
			\end{figure}
   
   \noindent In the remainder of the paper, we assume that $l$ is the line of symmetry if the parking nodes admit a single line of symmetry. If the parking nodes admit rotational symmetry, then $c$ is the center of rotational symmetry. We also assume that if the parking nodes admit rotational symmetry, then $l$ and $l'$ are perpendicular lines passing through $c$. Note that these lines divide the grid into four quadrants.
\subsection{Impossibility Results} 
In this subsection, we define all those initial configurations and the values of the capacities of the parking nodes for which the parking problem is unsolvable. First, we provide some definitions that are relevant in understanding the impossibility results.
\begin{lemma} \label{chap4-lemma1}
Let $\mathcal A$ be any algorithm for the parking problem in infinite grids. If there exists an execution of $\mathcal A$ such that the configuration $C(t)$ contains a robot multiplicity at a node that is not a parking node, then $\mathcal A$ cannot solve the parking problem.  
\end{lemma}
\begin{proof}
 Assume that the capacity of a parking node $p_i$ is $k_i$, where $k_i \geq 1$ and $i = 1, 2, 3, \ldots m$. So, without loss of generality, we assume that the capacity of each parking node is 1, i.e., in the final configuration, there must be exactly one robot at each parking node. Moreover, we assume that the scheduler is fully-synchronous, which implies that time is divided into global rounds, and all the robots are activated simultaneously. Suppose, at time $t >0$, a robot multiplicity is formed at one of the nodes, which is not a parking node. The robots at the multiplicity have the same local view of the configuration. If the adversary forces both the robots composing the multiplicity to perform the same move, the multiplicity in $C(t)$ is maintained for all $C(t')$, for $t'>t$. As a result, the robots in the multiplicity cannot reach different parking nodes. This proves that, under the execution of $\mathcal A$, the robots cannot reach the final configuration.  
\end{proof}
\noindent This lemma ensures that during the execution of any algorithm that solves the parking problem, the robots must perform a collision-less movement at all stages of the algorithm, unless the robot reaches its respective target parking node. Suppose the robots are oblivious and not endowed with global-strong multiplicity detection capability. In that case, they cannot detect whether exactly the $k_i$ number of robots reaches the parking node $p_i$. So, we formalize the result in the following lemma:
\begin{lemma}
Without the global-strong multiplicity detection capability of the robots, the parking problem is unsolvable.
\end{lemma}

\begin{lemma} \label{chap4-lemma3}
If the initial configuration $C(0) \in \mathcal I_{223}$ is such that the capacity of a parking node on $l$ is an odd integer. Then the parking problem is unsolvable.
\end{lemma}
\begin{proof}
Assume that there exists an algorithm $\mathcal A$ that solves the parking problem starting from any arbitrary initial configuration in $\mathcal I_{223}$ and the capacity of the parking node $p$ on $l$ is an odd integer. Without loss of generality, we assume that the capacity of $p$ is $2k+1$, where $k$ is a positive integer. Consider the scheduler to be semi-synchronous with the additional constraint that a robot and its symmetric image with respect to $l$ are activated and perform the Look-Compute-Move cycle simultaneously. In that case, a robot $r$ and its symmetric image $\phi(r)$ have the same local view of the configuration and they execute the same deterministic algorithm $\mathcal A$. Assume that at time $t>0$, there exists a $2k$ number of robots on $p$. Since the capacity of $p$ is $2k+1$, at time $t'>t$, a robot $r$ must start moving towards $p$. As the configuration is symmetric, there exists at least one execution of $\mathcal A$ out of different execution paths, where any move that $r$ performs according to $\mathcal A$ would result in a situation where $\phi(r)$ performs the same move. These movements of the robots ensure that at any moment of time, the configuration remains symmetric. Since there is no robot position on $l$ in the initial setup and all robots move in pairs, moving a robot $r$ to $l$ would also move $\phi(r)$ to the same node on $l$. While the robots $r$ and $\phi(r)$ move towards $p$, a multiplicity node is created at $p$, which cannot be separated by any deterministic algorithm. Hence, the parking problem is unsolvable, according to Lemma \ref{chap4-lemma1}. 
\end{proof}

Next, if $C(t)$ admits multiple lines of symmetry and if there exists a parking node on a line of symmetry with odd capacity, then the problem is unsolvable. We next state the following corollary of the Lemma \ref{chap4-lemma3} without going to the proof details.
\begin{corollary} \label{corolary1}
If the initial configuration $C(0) \in \mathcal I_{323}$, then the parking problem is unsolvable if the capacity of the parking node at $c$ is neither a multiple of 4 nor 2, depending on whether the angle of rotation is either $90 ^{\circ}$ or $180^{\circ}$.
\end{corollary}

\noindent Let $\mathcal U$ be the set of all configurations that are unsolvable according to Lemma \ref{chap4-lemma3} and Corollary \ref{corolary1}.

\section{Algorithm} \label{s4}
In this section, we propose a deterministic distributed algorithm for solving the parking problem in infinite grids. A sufficient condition for an initial configuration to be unsolvable, i.e., when the parking problem cannot be deterministically solved, is provided by the impossibility results given in Section \ref{s3}. The parking problem is solved using a deterministic distributed algorithm in this section for all initial configurations except from those indicated in Lemma \ref{chap4-lemma3} and Corollary \ref {corolary1}. The fundamental strategy of the proposed algorithm is to select a specific target parking node and permit a number of robots to move towards it, where the number of robots is equal to the parking node's capacity. The proposed algorithm mainly consists of the following phases:  \textit{Guard Selection and Placement (GS)} phase, \textit{Target Parking Node Selection (TPS)} phase, \textit{Candidate Robot Selection (CR)} phase,  \textit{Guard Movement (GM)} phase and \textit{Symmetry Breaking (SB) phase}. If the parking nodes are symmetric and the configuration is asymmetric, then in the GS phase, a guard is selected and moved in such a way that the configuration remains asymmetric during the execution of the algorithm. The robots identify the current configuration and determine whether a unique parking node could be selected for parking. If the configuration is symmetric, two parking nodes may be chosen for parking at any given time during the TPS phase. The closest number of robots equaling the capacity of the target parking node is selected in the CR phase and moves toward the target parking node in a sequential manner. While the parking node(s) with the highest orders become saturated, the next parking node is selected, which is unsaturated and has the highest order among all the parking nodes that are unsaturated. The process continues until each parking node becomes saturated. The guard moves toward its specific target parking node during the GM phase. The symmetric configurations that can be changed into asymmetric configurations are taken into consideration during the SB phase. A unique robot is selected and allowed to move toward an adjacent node such that the configuration becomes asymmetric. In the subsequent subsections, if the configuration is asymmetric or symmetric with a parking node on $l \cup \lbrace c \rbrace$, the robots can always select a unique parking node in order to begin the parking formation process.

\subsection{Ordering of the parking nodes}
In this subsection, we first consider all those configurations, where the parking nodes of the configurations can be ordered uniquely. This ordering is necessary to identify a unique parking node, which will be selected by the robots in order to initialize the parking formation. So, first, consider the case when the parking nodes are asymmetric. According to the definition of the symmetricity of the set $\mathcal P$, there exists a unique lexicographic string $\alpha_i$. As a result, there exists a unique leading corner. Consider the string representation of the string $\alpha_i$ associated to the unique leading corner. We define an ordering $\mathcal O$ of the parking nodes as follows: Consider the ordering of the parking nodes that appears from first to last in the string representation of the string $\alpha_i$ associated to the unique leading corner. Assume that $(p_1, p_2, \ldots, p_m)$ be the respective order of the parking nodes that appear in the string representation of the string $\alpha_i$ associated to the unique leading corner. This particular ordering of the parking nodes is defined as $\mathcal O_1$. Since the parking nodes are fixed nodes located at the nodes of the grid and the ordering $\mathcal O_1$ is defined with respect to the position of the parking nodes, the ordering $\mathcal O_1$ remains invariant. 

\noindent Next, consider the case when the parking nodes are symmetric with respect to a unique line of symmetry $l$ and there exist parking nodes on $l$. In this case, there may exist one or two leading corners. In case there exists a unique leading corner, then $l$ must be a diagonal line of symmetry. Note that both the strings associated with the unique leading corner are equal. Here, one of the directions is arbitrarily selected as the string direction. Consider the string representation of the string $\alpha_i$ associated to the unique leading corner. We define an ordering $\mathcal O_2$ of the parking nodes as follows: Consider the ordering of the parking nodes on $l$ that appears from first to last in the string representation of the string $\alpha_i$ associated to the unique leading corner. Assume that $(p_1, p_2, \ldots, p_z)$ be the respective order of the parking nodes on $l$ that appear in the string representation, where $z$ denotes the number of parking nodes on $l$. This particular ordering of the parking nodes is defined as $\mathcal O_2$. Next, assume the case when there exist two leading corners. Consider the string representation of the string $\alpha_i$ associated to the leading corners. In this case, the ordering $\mathcal O_2$ is defined similarly to before. Hence, we have the following observations.
\begin{observation}
If the parking nodes are asymmetric, then they can be ordered.
\end{observation}
\begin{observation}
If the parking nodes are symmetric with respect to a unique line of symmetry $l$, then the parking nodes on $l$ are orderable.
\end{observation}

\subsection{Guard Selection and Placement $(GS)$}
 
 Consider the case when the parking nodes are symmetric, but the configuration is asymmetric. In this phase, a unique robot is selected as a guard and placed in such a way that the configuration remains asymmetric during the execution of the algorithm. The following notations are used in describing this phase:
\begin{itemize}
   
    \item Condition $C_1$: There exists at least one robot position outside the rectangle $M_{\mathcal P}$.
    \item Condition $C_2$: Each robot is inside the rectangle $M_{\mathcal P}$.
    \item Condition $C_3$: There exists a unique farthest robot from $l \cup \lbrace c \rbrace$.
\end{itemize}
Depending on the class of configurations to which $C(t)$ belongs, the phase is described in Table \ref{tab:guard}. If there are more than one furthest robot from the key corner, then since the configuration is asymmetric, a unique robot can always be selected according to the view of the robots. Note that while the guard is selected and placed, the guard is the unique farthest robot from $l \cup \lbrace c \rbrace$. As a result, it does not have any symmetric image with respect to $l \cup \lbrace c \rbrace$, which implies that the configuration remains asymmetric during the execution of the algorithm. 
\begin{table}
\centering
\begin{tabular}{ |p{3.4cm}| p{5.7cm}|p{5.7cm}|}
 \hline
 \multicolumn{3}{|c|}{\textbf{Guard Selection and Placement}} \\
 \hline
 \textbf{Initial Configuration ($\boldsymbol{\mathcal I_{21} \cup \mathcal I_{31}}$}) & \textbf{Guard} & \textbf{Position of the guard } \\
 \hline
  $ C_1 \land C_3$ & The unique robot farthest from $l \cup \lbrace c \rbrace$ & Current position of the guard\\
 \hline
  $C_1 \land \lnot C_3 $ & The unique robot furthest from $l \cup \lbrace c \rbrace$ and having the maximum configuration view among all the furthest robots & The unique robot moves towards an adjacent node away from $l \cup \lbrace c \rbrace$\\
 \hline
  $C_2 \land C_3 $ & The unique robot furthest from $l \cup \lbrace c \rbrace$ & The guard continues its movement away from $l \cup \lbrace c \rbrace$, unless the condition $C_1$ becomes true \\
 \hline 
  $C_2 \land \lnot C_3 $ & The unique robot furthest from $l \cup \lbrace c \rbrace$ and having the maximum configuration view among all the furthest robots & The guard continues its movement away from $l \cup \lbrace c \rbrace$ until the condition $C_1$ becomes true \\

 \hline
 
\end{tabular}
\caption{\label{tab:guard}Guard Selection and Placement}
\end{table}

\subsection{Half-planes and Quadrants}
First, consider the case when $C(0) \in \mathcal I_{21}$. The line of symmetry $l$ divides the entire grid into two half-planes. We consider the open half-planes, i.e., the half-planes excluding the nodes on $l$. Let $H_1$ and $H_2$ denote the two half-planes delimited by $l$. The following definitions are to be considered.
\begin{enumerate}
    \item $UP(t)$- Number of parking nodes which are unsaturated at time $t$.
    
    \item \textit{Deficit Measure of a parking node $p_i$ ($Df_{p_i}(t)$):} The deficit measure $Df_{p_i}(t)$ of a parking node $p_i$ at time $t$ is defined as the deficit in the number of robots needed to have exactly $k_i$ robots on $p_i$.
    
    \item $K_1=$ $\sum\limits_{p_i\in H_1} Df_{p_i}(t)$ denotes the total deficit in order to have exactly $\sum\limits_{p_i\in H_1} k_i$ number of robots at the parking nodes belonging to the half-plane $H_1$.
     \item $K_2=$ $\sum\limits_{p_i\in H_2} Df_{p_i}(t)$ denotes the total deficit in order to have exactly $\sum\limits_{p_i\in H_2} k_i$ number of robots at the parking nodes belonging to the half-plane $H_2$.
      \end{enumerate}
     \begin{definition}
          Let $C(t)$ be any initial configuration belonging to the set $\mathcal I_{21}$. $C(t)$ is said to be \textit{unbalanced} if the two half-planes delimited by $l$ contain an unequal number of robots. Otherwise, the configuration is said to be \textit{balanced}.
     \end{definition}
    \noindent We next consider the following conditions.
    
     \begin{enumerate}
    \item Condition $C_4$- There exists a unique half-plane that contains the minimum number of unsaturated parking nodes.
    \item Condition $C_5$- $K_1 \neq K_2$
    \item Condition $C_6$- The configuration is unbalanced.
    \item Condition $C_7$- The configuration is balanced and $\mathcal R \cap l \neq \emptyset$.
    \item Condition $C_8$- The configuration is balanced and $\mathcal R \cap l = \emptyset$. 
\end{enumerate}
The half-plane $\mathcal H_{target}$ or $\mathcal H^{+}$ is defined according to Table \ref{tab:half}, where the parking at the parking nodes initializes. The other half-plane is denoted by $\mathcal H^{-}$. In Figure \ref{halfplane1} (a), $ABCD$ is the $M_{\mathcal P}$ and $AB'C'D$ is the $MER$. Assume that the capacities of the parking nodes $p_1$, $p_2$, $p_3$ and $p_4$ are 2, 2, 1 and 1, respectively. The half-plane with more number of robots is selected as $\mathcal H^{+}$. In Figure \ref{halfplane1} (b), assume that the capacities of the parking nodes $p_1$, $p_2$, $p_3$ and $p_4$ are 3, 3, 2 and 2, respectively. Each of the half-planes contains the same number of robots. Therefore, the configuration is balanced. The half-plane not containing the guard $r_5$ is defined as $\mathcal H^{+}$. 
\begin{table}
\centering
\begin{tabular}{ |p{5cm}|p{9.8cm}|  }
\hline
\multicolumn{2}{|c|}{\textbf{Demarcation of the half-planes for fixing the target}} \\
\hline
\multicolumn{1}{|c|}{\textbf{Initial Configuration} $\boldsymbol{(\mathcal I_{21}}$)} & \multicolumn{1}{|c|}{$\boldsymbol{\mathcal H^{+}}$}\\ 
\hline
{\small $C_4$ }  & T{\small he unique half-plane which contains the minimum number of unsaturated parked nodes} \\
\hline
{\small $\lnot$ $C_4$ $\land$ $C_5$ $\land$ $K_1< K_2$} & {\small $H_1$}\\ 
\hline
{\small $\lnot$ $C_4$ $\land$ $C_5$ $\land$ $K_2< K_1$} & {\small $H_2$}\\
\hline
{\small $\lnot$ $C_4$ $\land$ $\lnot C_5$ $\land$ $C_6$} & {\small The unique half-plane with the maximum number of robot positions}\\
\hline
{\small $\lnot$ $C_4$ $\land$ $\lnot C_5$ $\land$ $\lnot$ $C_6$ $\land$ $C_7$} & {\small The northernmost robot on $l$ move towards an adjacent node away from $l$. The unique half-plane with the maximum number of robot positions is defined as $\mathcal H^{+}$}\\
\hline
{\small $\lnot$ $C_4$ $\land$ $\lnot C_5$ $\land$ $\lnot$ $C_6$ $\land$ $\lnot$ $C_7 \land C_8 $} & {\small The unique half-plane not containing the guard}\\
\hline
\end{tabular}
\caption{\label{tab:half}Demarcation of the half-planes}
\end{table}%

\begin{figure}[h]
			\centering
				\subfloat[]
			{
				\includegraphics[width=0.37\columnwidth]{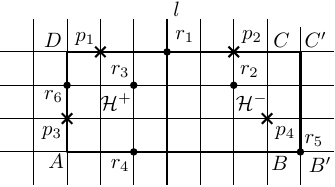}
			}
   \hspace{0.5cm}
   \subfloat[]
			{
				\includegraphics[width=0.37\columnwidth]{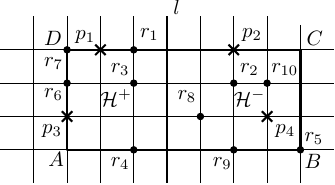}
			}
			
			\caption{Example configuration showing demarcations of half-planes.} 
			\label{halfplane1}
			\end{figure}

\noindent Next, consider the lines $l$ and $l'$ that pass through the center of $MER$. These lines divide the entire rectangle into four quadrants. Let $Q_i$ be the quadrants that are defined by the lines $l$ and $l'$, where $i$ ranges from 1 to 4. We consider the open quadrants, i.e., the quadrants exclude the nodes belonging to $l$ and $l'$. Let $L_j=$ $\sum\limits_{p_i\in Q_j} Df_{p_i}(t)$ denote the total deficit in order to have exactly $\sum\limits_{p_i\in Q_j} k_i$ number of robots at the parking nodes belonging to the quadrant $Q_j$, where $j$ ranges from 1 to 4.

 \begin{table}
 \centering
 \begin{tabular}{ |p{5.2cm}|p{9.9cm}|  }
 \hline
 \multicolumn{2}{|c|}{\textbf{Demarcation of the quadrants for fixing the target}} \\
 \hline
 \multicolumn{1}{|c|}{\textbf{Initial Configuration} $\boldsymbol{C(0)}$} & \multicolumn{1}{|c|}{$\boldsymbol{\mathcal Q_{target}}$}\\ 
 \hline
 {\small $C_9$ }  & T{\small he unique quadrant which contains the minimum number of unsaturated parking nodes} \\
 \hline
 {\small $\lnot$ $C_9$ $\land$ $C_{10}$ $\land$ $L_j$ is minimum} & {\small $Q_j$}\\ 
 \hline

 {\small $\lnot$ $C_9$ $\land$ $\lnot C_{10}$ $\land$ $C_{11}$} & {\small The unique quadrant with the maximum number of robot}\\
 \hline
 {\small $\lnot$ $C_9$ $\land$ $\lnot C_{10}$ $\land$ $\lnot$ $C_{11}$ $\land$ $C_{12}$} & {\small The robot on $l$ or $l'$ with the highest order moves along $l \cup l'$ and when it is one move away from a quadrant containing a maximum number of robot positions, it moves towards an adjacent node away from $l \cup l'$. The unique quadrant with the maximum number of robot is $\mathcal Q_{target}$}\\
 \hline
 {\small $\lnot$ $C_9$ $\land$ $\lnot C_{10}$ $\land$ $\lnot$ $C_{11}$ $\land$ $\lnot$ $C_{12} \land C_{13} $} & {\small In case there are two adjacent quadrants containing the minimum $L_j$ value and the maximum number of robot positions, the unique quadrant that does not contain the guard. If the quadrants containing the minimum $L_j$ value and the maximum number of robot positions are such that they are non-adjacent, then the quadrant that is not adjacent to the quadrant containing the guard is $\mathcal Q_{target}$ }\\
 \hline
 \end{tabular}
 \caption{\label{tab:quadrant1}Demarcation of the quadrants}
 \end{table}
 \begin{definition} \label {def1}
      Let $C(t)$ be any initial configuration belonging to the set $\mathcal I_{31}$. Consider the quadrants with the minimum $L_j$ value. $C(t)$ is said to be \textit{unbalanced} if there exists a unique quadrant with minimum $L_j$ value and such that it contains the maximum number of robots. Otherwise, the configuration is said to be \textit{balanced}.   
     \end{definition} 
     \noindent According to the Definition \ref{def1}, a configuration is said to be \textit{balanced} if it contains at least two quadrants with minimum $L_j$ value and with the maximum number of robot positions. We next consider the following conditions.
     \begin{enumerate}
     \item Condition $C_9$- There exists a unique quadrant that contains the minimum number of unsaturated nodes.
     \item Condition $C_{10}$- There exists a unique quadrant with the minimum $L_j$ value, for $j=\lbrace 1, 2, 3, 4 \rbrace$.
     \item Condition $C_{11}$- The configuration is unbalanced.
     \item Condition $C_{12}$- The configuration is balanced and ($\mathcal R \cap l \neq \emptyset$ $\lor$ $\mathcal R \cap l' \neq \emptyset$).
    \item Condition $C_{13}$- The configuration is balanced and $\mathcal R \cap l = \emptyset$, $\mathcal R \cap l' \neq \emptyset$. 
     
 \end{enumerate}
 Note that $\lnot C_{10}$ implies there exist at least two quadrants with the minimum $L_j$ value. The quadrant $\mathcal Q_{target}$ or $\mathcal Q^{++}$, where the parking is initialized, is defined according to Table \ref{tab:quadrant1}. The quadrants adjacent to $\mathcal Q^{++}$ with respect to $l$ and $l'$ are defined as $\mathcal Q^{-+}$ and $\mathcal Q^{+-}$ respectively. Similarly, the quadrant non-adjacent to $\mathcal Q^{++}$ is defined as $\mathcal Q^{--}$. In Figure \ref{quadrant1}, assume that the capacities of each of the parking nodes are 2. As a result, each $L_j$ equals 2 in the initial configuration. The quadrant with more number of robots is defined as $\mathcal Q^{++}$. Consider the case when there are exactly two adjacent quadrants with the minimum $L_j$ value and with the maximum number of robots. Suppose one of such quadrants contains the guard. In that case, the quadrant not containing the guard is defined as $\mathcal Q^{++}$.

\begin{definition}
Consider the half-lines starting from $c$ and along the grid-lines. We denote the wedge boundaries of the quadrants delimited by the lines $l$ and $l'$ by $\mathcal B_i$, where $i=\lbrace 1, 2, 3, 4 \rbrace$.
\end{definition}

\begin{figure}[h]
			\centering
				
			{
				\includegraphics[width=0.320\columnwidth]{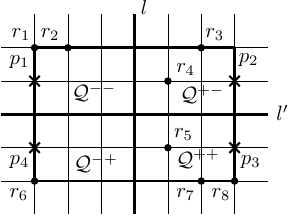}
			}

			\caption{Example configuration showing demarcations of quadrants.} 
			\label{quadrant1}
			\end{figure}
 
 \subsection{Target Parking Node Selection (TPS)}
In this phase, the target parking node for the parking problem is selected. Depending on the following classes of configurations, the phase is described in Table \ref{tab:target}. Let $p_{guard}$ be the closest parking node from the guard. If multiple such parking nodes exist, the parking node closest to the guard and having maximum order is selected as $p_{guard}$. We first assume that the target parking nodes are selected in $\mathcal P \setminus \lbrace p_{guard} \rbrace$. While all the parking nodes belonging to the set $\mathcal P \setminus \lbrace p_{guard} \rbrace$ become saturated, $p_{guard}$ becomes the target parking node. We next consider the following conditions that are relevant in understanding this phase. 
\begin{enumerate}

\item $C_{14}$- There exists an unsaturated parking node on $l$.
\item $C_{15}$- There exists an unsaturated parking node on $c$.
 \item $C_{16}$- All the parking nodes belonging to $\mathcal H^{+}$ are saturated.
 \item $C_{16}'$- All the parking nodes belonging to $\mathcal H^{-}$ are saturated.
 
      \item $C_{17}$- All the parking nodes belonging to $\mathcal Q^{++}$ are saturated. 
  \item $C_{18}$- All the parking nodes at the wedge boundaries corresponding to the quadrant $\mathcal Q^{++}$ are saturated.
  \item $C_{19}$- All the parking nodes belonging to the quadrant $\mathcal Q^{-+}$ are saturated.

 \item $C_{20}$- All the parking nodes at the wedge boundaries corresponding to the quadrant $\mathcal Q^{-+}$ are saturated. 
  \item $C_{21}$- All the parking nodes belonging to the quadrant $\mathcal Q^{+-}$ are saturated.
 \item $C_{22}$- All the parking nodes at the wedge boundaries corresponding to the quadrant $\mathcal Q^{+-}$ are saturated. 
 \item $C_{23}$- All the parking nodes belonging to $\mathcal Q^{--}$ are saturated.
\end{enumerate}
Note that $\lnot C_{14}$ implies that the parking nodes are symmetric with respect to $l$ and there either does not exist any parking node on $l$ or each parking node on $l$ is saturated. In Figure \ref{halfplane1}, $A$ and $B$ are the leading corners. $p_1$ is the parking node in $\mathcal H^{+}$ which has the highest order. The target parking nodes are selected in the order ($p_1$, $p_3$, $p_2$, $p_4$).

\begin{table}
\centering
\begin{tabular}{ |p{4.6cm}|p{9.9cm}|  }
\hline
\multicolumn{2}{|c|}{\textbf{Target Parking Node Selection}} \\
\hline
\multicolumn{1}{|c|}{\textbf{Initial Configuration} $\boldsymbol{C(0)}$} & \multicolumn{1}{|c|}{\textbf{Target Parking Node}}\\ 
\hline
{\small $\mathcal I_1$ }  & T{\small he parking node which is unsaturated and has the highest order with respect to $\mathcal O_1$ } \\
\hline
{\small $\mathcal I_{2}$ $\land$ $C_{14}$} & {The parking node on $l$ which is unsaturated and has the highest order with respect to $\mathcal O_2$}\\ 
\hline

{\small $\mathcal I_{21}$ $\land$ $\lnot$ $C_{14}$ $\land$ $\lnot$ $C_{16}$} & {\small The parking node, which is unsaturated and has the highest order in $\mathcal H^{+}$ among all the unsaturated nodes in $\mathcal H^{+}$}\\
\hline
{\small $\mathcal I_{21}$ $\land$ $\lnot$ $C_{14}$ $\land$ $C_{16}$ $\land$ $\lnot C_{16}'$} & {\small The parking node, which is unsaturated and has the highest order in $\mathcal H^{-}$ among all the unsaturated nodes in $\mathcal H^{-}$}\\
\hline

{\small $\mathcal I_{22}$ $\land$ $\lnot$ $C_{14}$} & {\small The two parking nodes that have the highest order among all the unsaturated parking nodes and lying on two different half-planes }\\
\hline
{$\mathcal I_3$ $\land$ $C_{15}$}& {\small The parking node on $c$}\\
 \hline
 {$\mathcal I_{31}$ $\land  \lnot C_{15}$ $\land \lnot C_{17}$}& {\small The parking node in $\mathcal Q^{++}$ which is unsaturated and has the highest order among all the unsaturated parking nodes in $\mathcal Q^{++}$} \\
 \hline
 {$\mathcal I_{31}$ $\land  \lnot C_{15}$ $\land C_{17}$ $\land$ $\lnot C_{18}$}& {\small The parking node at the wedge boundaries corresponding to $\mathcal Q^{++}$ which is unsaturated and has the highest order among such nodes} \\

 \hline
 {$\mathcal I_{31}$ $\land  \lnot C_{15}$ $\land  C_{17} \land C_{18 } \land \lnot C_{19} $}& {\small The parking node in $\mathcal Q^{-+}$ which is unsaturated and has the highest order among all the unsaturated parking nodes in $\mathcal Q^{-+}$} \\
 \hline 
 { \small $\mathcal I_{31}$ $\land  \lnot C_{15}$ $\land  C_{17} \land C_{18 } \land C_{19} \land \lnot C_{20}$}& {\small The parking node at the wedge boundaries corresponding to $\mathcal Q^{-+}$ which is unsaturated and has the highest order among such unsaturated parking nodes} \\
 \hline
 { \small $\mathcal I_{31}$ $\land  \lnot C_{15}$ $\land  C_{17} \land C_{18 } \land C_{19} \land  C_{20} \land \lnot C_{21}$}& {\small The parking node in $\mathcal Q^{+-}$ which is unsaturated and has the highest order among all the unsaturated parking nodes in $\mathcal Q^{+-}$} \\
 \hline
 { \small $\mathcal I_{31}$ $\land  \lnot C_{15}$ $\land  C_{17} \land C_{18 } \land C_{19} \land  C_{20} \land C_{21} \land \lnot C_{22}$}& {\small The parking node at the wedge boundaries corresponding to $\mathcal Q^{+-}$ which is unsaturated and has the highest order among such unsaturated parking nodes} \\
 \hline
 { \small $\mathcal I_{31}$ $\land  \lnot C_{15}$ $\land  C_{17} \land C_{18 } \land C_{19} \land  C_{20} \land C_{21} \land C_{22} \land \lnot C_{23}$}& {\small The parking node in $\mathcal Q^{--}$ which is unsaturated and has the highest order among all the unsaturated parking nodes in $\mathcal Q^{--}$} \\
 
 \hline
 {\small $ \mathcal I_{32}$ $\land  \lnot C_{15}$} &{\small The two or four parking nodes which have the highest order among all the unsaturated nodes and lying on different quadrants or on the wedge boundaries }\\
 \hline
\end{tabular}
\caption{\label{tab:target}Target Parking Node Selection}
\end{table}
\subsection{Candidate Robot Selection Phase}
In view of Lemma \ref{chap4-lemma1}, while a robot moves towards a parking node, it must ensure collision-free movement. Otherwise, the problem becomes unsolvable. As a result, a robot will move toward its target only when it has a path toward that target that does not contain any other robot positions. 
\begin{definition}
A path from a robot to a parking node is said to be \textit{free} if it does not contain any other robot positions.
\end{definition}

\noindent A robot would move toward its target only when it has a free path toward it. In this phase, the \textit{candidate robot} is selected and allowed to move toward the target parking node. Let $p \neq p_{guard}$ be the target parking node selected in the TPS phase. Depending on the different classes of configurations, the following cases are to be considered.

\begin{enumerate}
    \item $C(t)$ is asymmetric. As a result, the robots are orderable. The robot that does not lie on any saturated parking node and has the shortest free path to $p$ is selected as the candidate robot. If multiple such robots exist, the one with the highest order among such robots is selected as the candidate robot. 

    \item $C(t)$ is symmetric with respect to a single line of symmetry $l$. First, assume that $p$ is on $l$. If at least one robot exists on $l$, then the symmetry can be broken by allowing a robot from $l$ to move towards an adjacent node away from $l$. As a result, assume that there is no robot position on $l$. The two closest robots, which do not lie on any saturated parking node and have shortest free paths towards $p$, are selected as the candidates for $p$. If there are multiple such robots, the ties are broken by considering the robots that lie on different half-planes and have the highest order among all such robots. Next, assume that $p$ is on the half-planes. The robot that does not lie on any saturated parking node and has a shortest free path toward $p$ is selected as the candidate robot. Note that such candidates are selected in both half-planes.
     \item $C(t)$ is symmetric with respect to rotational symmetry. First, assume that $p$ is on $c$. If there exists a robot on $c$, the robot on $c$ moves towards an adjacent node, and the configuration becomes asymmetric. Assume the case when there are no robots on $c$. The robots that are closest to $p$ are selected as candidate robots. In this case, depending on whether the angle of rotational symmetry is $180 ^{\circ}$ or $90 ^{\circ} $, two or four robots are selected as candidates. Next, assume that $p$ is located either on a quadrant or on of the wedge boundaries. If the target parking node lies on a quadrant, the robot that does not lie on any saturated parking node and has a shortest free path toward $p$ is selected as the candidate robot. It should be noted that such candidates are chosen from each of the four quadrants, for each target parking node. Otherwise, if the target parking node is on a wedge boundary, the robot(s) not lying on any saturated parking node and having a shortest free path towards the target is (are) selected as candidate robot(s). 
\end{enumerate}

\noindent Next, assume that $p_{guard}$ is the target parking node. The candidates are selected as the robot which has shortest free path towards $p_{guard}$. Finally, the guard moves towards $p_{guard}$. By the choice of $p$, there always exists a half-line starting from $p$, which does not contain any robot position. As a result, a free path always exists between the candidate robot and $p$.

\subsection{Guard Movement}
Assume the case when the parking nodes are symmetric and the configuration is asymmetric. In the GM phase, the guard moves toward its respective destination and the parking process is terminated. The guard moves only when it finds that, except for one, all the parking nodes have become saturated. It moves towards its destination $p$ in a free path. The guard moves towards its destination and each parking node becomes saturated, transforming the configuration into a final configuration.
\subsection{Symmetry Breaking Phase}
In this phase, the symmetric configurations that can be transformed into asymmetric configurations are considered. The list of all configurations that are considered in this phase is as follows:
\begin{enumerate}
    \item Configurations admitting a single line of symmetry $l$ with at least one robot position on $l$, i.e., $C(0) \in \mathcal I_{221}$.
    \item Configurations admitting rotational symmetry with a robot on the center of rotational symmetry, i.e., $C(0) \in \mathcal I_{331}$.
\end{enumerate}
Let $r$ be the robot on $l$ with the maximum order, in case $C(0) \in \mathcal I_{221}$, i.e., $r$ appears after every other robot that is on $l$ in the string directions associated to the leading corner(s). In case, $C(0) \in \mathcal I_{331}$, let $r$ be the robot on $c$. It should be noted that while the robot $r$ on $l \cup \lbrace c \rbrace$ moves towards an adjacent node away from $l \cup \lbrace c\rbrace$, the configuration transforms into an asymmetric configuration. However, it might be possible that the neighbors of $r$ contain robot positions. The movement of the robot towards an adjacent node might create a robot multiplicity node, resulting in an unsolvable configuration according to Lemma \ref{chap4-lemma1}. As a result, there must be a free space available around $r$, so that it can move toward an adjacent node. Due to the asynchronous behavior of the scheduler, there might be a possible pending move while the adjacent robots of $r$ move towards an adjacent node away from $l \cup \lbrace c \rbrace$. Therefore, we consider the following definitions.
\begin{definition}
\noindent A function $h_t:V\rightarrow\lbrace 0, 1 \rbrace$ at any fixed time $t\geq 0$ is defined by:
 \[ h_t(v)=\begin{cases} 
 0 & \text{if} \;v \textrm {  is a free node}  \\
 1 & \text{if} \;v \textrm{ contains a robot position } 
 \end{cases}
 \]
\end{definition}

\noindent Note that $h_t(v)$ is defined as an indicator variable on the set of nodes $V$ which equals 1 when the node contains a robot's position at time $t$.
\begin{definition}
	Let $u$ be a node of the input graph containing a robot position $r$. Consider all the half-lines starting from the node $u$ containing the robot position $r$. Scan all those half-lines and associate $h_t(v)$ to each node $v$, that the half-line encounters. A string terminates when the last node occupied by a robot in the half-line is encountered. The four strings that are generated by the robot $r$ (\cite{DBLP:conf/ictcs/Cicerone20}) are denoted by $\beta_{left}(r)$, $\beta_{right}(r)$, $\beta_{top}(r)$ and $\beta_{bottom}(r)$.
\end{definition}
Next, assume that the initial configuration $C(0)$ is symmetric with respect to rotational symmetry. Let $r$ be the robot at the center of rotation $c$. Then, we have the following definition.
\begin{definition}
	 Consider all the strings generated by the robot $r$. The strings are said to be \textit{nearly equal} if the adjacent nodes of $c$ are occupied by the robots and can be obtained from one another by reversing one occurrence of the substring $01$. That is, the strings can be made equal by moving a robot towards an adjacent node whose movement is pending. A configuration is said to be \textit{nearly rotational}, if the strings generated by the central robot $r$ are \textit{nearly equal}.
\end{definition}
Next, assume that the initial configuration $C(0)$ is symmetric with respect to a unique line of symmetry $l$. Let $r$ be the unique robot on $l$ that has the maximum configuration view.
\begin{definition}
	 Consider the strings $\beta_{left}(r)$ and $\beta_{right}(r)$ as the strings generated by $r$ and which terminate away from $l$. The strings $\beta_{left}(r)$ and $\beta_{right}(r)$ are said to be \textit{nearly equal}, if the strings  $\beta_{left}(r)$ and $\beta_{right}(r)$ can be obtained from one another by just reversing one occurrence of the substring $01$. A configuration is said to be \textit{nearly reflective}, if the strings $\beta_{left}(r)$ and $\beta_{right}(r)$ are \textit{nearly equal}.
\end{definition}
\begin{figure}[h]
	\centering
	\hspace{0.1cm}
{
		\includegraphics[width=0.20\columnwidth]{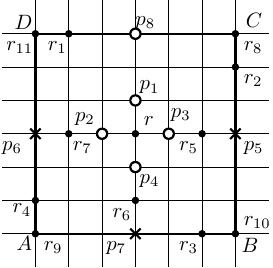}
	}
	\hspace*{0.65cm}
{
		\includegraphics[width=0.24\columnwidth]{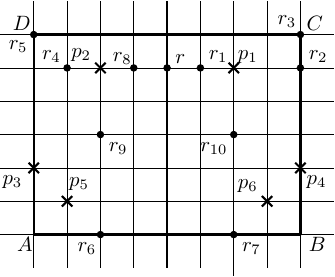}
	}
	\caption{(a) \textit{Nearly rotational configuration}. (b) \textit{Nearly reflective configuration}. }
	\label{strings}
\end{figure}
\noindent In Figure~\ref{strings}(a), the \textit{strings generated by the robot $r$} are given by $\lbrace 110,110,110,101\rbrace$. The circles at $p_1$, $p_2$, $p_3$, $p_4$ and $p_8$ denotes parking nodes with a robot position on it. In Figure~\ref{strings}(b), the strings generated by the robot $r$ and which terminate away from $l$ are given by $\lbrace 1010, 1001\rbrace$. 
\section{Correctness} \label{chap6:sec6}
\begin{lemma}
    In the GS phase, the guard remains invariant while it moves towards its destination.
\end{lemma}
\begin{proof}
	 	First, assume that the parking nodes are symmetric with respect to $l$ and there do not exist any parking nodes on $l$. The proof follows similarly when the configuration admits rotational symmetry, and there does not exist any parking node on $c$. The following cases are to be considered.
	 	\setcounter{case}{0}
	 	\begin{case} \label{case111}
	 	\normalfont
	 	There exists at least one robot position outside the rectangle $M_{\mathcal P}$. Note that there may be multiple such robots. If there exists precisely one such robot $r$, then in the GS phase, $r$ is selected as a guard. Otherwise, if there are multiple such robots, the unique robot $r$ with the maximum configuration view among such robots is selected as the guard. $r$ moves towards an adjacent node $v$ away from $l$. The moment it reaches $v$, it becomes the unique farthest robot from $l$. As the guard is selected as the unique farthest robot from $l$, it remains uniquely identifiable by the other robots.  
	 	\end{case}
	 	\begin{case}
	 	\normalfont
	 	Each robot position is inside or on the rectangle $M_{\mathcal P}$. In that case, the robot position farthest from $l$ is selected as a guard. Note that there may be multiple such robots. The GS phase ensures that the guard is selected as the unique robot $r$, which is farthest from $l$ and with the maximum configuration view in case of a tie. The moment $r$ moves towards an adjacent node away from $l$, it becomes the unique farthest robot from $l$. $r$ continues its movement unless it becomes the unique robot outside the rectangle $M_{\mathcal P}$. The rest of the proof follows from Case \ref{case111}.
	 	\end{case}
\end{proof}
\begin{lemma}
    During the execution of the algorithm Parking(), if the parking nodes admit a single line of symmetry $l$, then $\mathcal H^{+}$ remains invariant. \label{chap6-lemma2}
\end{lemma}
\begin{proof}
    The following cases are to be considered.
   \setcounter{case}{0}
        \begin{case} \label{proof1-case1}
        \normalfont
            Condition $C_4$ holds. This implies that there exists a unique half-plane $\mathcal H^{+}$, which contains the minimum number of unsaturated parking nodes. As a result, there exists at least one saturated parking node in $\mathcal H^{+}$. The target parking node is selected in $\mathcal H^{+}$ in the TPS phase, which has the highest order among all unsaturated parking nodes in $\mathcal H^{+}$. During the execution of the algorithm, no robot at the saturated parking nodes is allowed to move. As a result, a saturated parking node can never become unsaturated, and the half-plane with the minimum number of unsaturated parking nodes remains invariant. Hence, $\mathcal H^{+}$ remains invariant. 
\end{case}
\begin{case} \label{proof1-case2}
\normalfont
    Condition $ \lnot C_4 \land C_5$ holds. This implies that the number of unsaturated parking nodes is the same in both the half-planes. However, $K_1 \neq K_2$. Assume that $K_1 <K_2$. The target parking node is selected in $\mathcal H^{+}$ with the minimum $K_i$ value, $i \in \lbrace 1,2\rbrace$, i.e., the half-plane, which has the minimum total deficit measure. Suppose at time $t>0$, a candidate robot is selected in the CR phase and allowed to move towards the target parking node. While this robot reaches the target parking node, the value of $K_1$ becomes much less than $K_2$. Eventually, there exists a time $t'> t$ in which at least one parking node in $\mathcal H^{+}$ becomes saturated. The rest of the proof follows from Case \ref{proof1-case1}.  
\end{case}
\begin{case} \label{proof1-case3}
\normalfont
    Condition $\lnot$ $C_4$ $\land$ $\lnot C_5$ $\land$ $C_6$ holds. This implies that $K_1=K_2$ and the configuration is unbalanced. $\mathcal H^{+}$ is selected as the half-plane with the maximum number of robots. The symmetry of the parking nodes is also defined with respect to their capacities. As a result, according to the execution of the algorithm, a robot in $\mathcal H^{+}$ will only move to the half-plane $\mathcal H^{-}$ when all of the parking nodes in $\mathcal H^{+}$ become saturated. There will eventually be a time $t>0$, when a robot reaches the target parking node, resulting in the value of $K_i$ being less than the value of $K_j$, $i, j= \lbrace 1, 2 \rbrace$. The rest of the proof follows from Case \ref{proof1-case2}.
\end{case}
\begin{case}
\normalfont
    Condition $\lnot$ $C_4$ $\land$ $\lnot C_5$ $\land$ $\lnot$ $C_6$ $\land$ $C_7$ holds. This implies the configuration is balanced and at least one robot position exists on $l$. The robot on $l$ with the maximum configuration view moves towards an adjacent node and the configuration becomes unbalanced, resulting in the condition $C_6$ evaluating to true. The rest of the proof follows from Case \ref{proof1-case3}. 
\end{case}
\begin{case}
\normalfont
    Condition $\lnot$ $C_4$ $\land$ $\lnot C_5$ $\land$ $\lnot$ $C_6$ $\land$ $\lnot C_7$ $\land C_8$ holds. This implies the configuration is balanced and there do not exist any robot positions on $l$. $\mathcal H^{+}$ is selected as the half-plane that does not contain the guard. The guard moves only when each parking node except for one becomes saturated. As a result, the target parking node will be selected in $\mathcal H^{+}$, unless each parking node in $\mathcal H^{+}$ becomes saturated. As a result, eventually at some time the condition $C_5$ will hold true and $\mathcal H^{+}$ will remain invariant from that instant of time.
\end{case}

\end{proof}
\begin{lemma}
  During the execution of the algorithm Parking(), if the parking nodes admit rotational symmetry, then $\mathcal Q^{++}$ remains invariant. \label{chap6-lemma3}
\end{lemma}
\begin{proof}
 The following cases are to be considered.
 \setcounter{case}{0}
 \begin{case} \label{proof2-case1}
 \normalfont
 Condition $C_9$ holds. This implies that there exists a unique quadrant $\mathcal Q^{++}$, that contains the minimum number of unsaturated parking nodes. As a result, there exists at least one saturated parking node in $\mathcal Q^{++}$. The target parking node is selected in $\mathcal Q^{++}$ in the TPS phase, which has the highest order among all unsaturated parking nodes in $\mathcal Q^{++}$. Since no robot lying on a saturated parking node is allowed to move in the CRS phase, a saturated parking node can never become unsaturated. This implies that the minimum number of unsaturated parking nodes in the quadrant remains invariant. Hence, $\mathcal Q^{++}$ remains invariant.
 \end{case}
 \begin{case} \label{proof2-case2}
 \normalfont 
 Condition $ \lnot$ $C_9 \land C_{10}$ holds. This implies that the number of unsaturated parking nodes in each quadrant is the same. However, there exists a unique quadrant with a minimum $L_j$ value, for some $j \in \lbrace 1,2,3,4 \rbrace$, i.e., there exists a unique quadrant with the minimum total deficit measure. Assume that the quadrant $Q_j$ has the minimum $L_j$ value. The target parking node is selected in $\mathcal Q_j$, which is demarcated as $\mathcal Q^{++}$. Suppose at time $t>0$, a candidate robot is selected in the CRS phase and allowed to move towards the target parking node in $\mathcal Q^{++}$. While this robot reaches the target parking node, the value of $L_j$ decreases further. As a result, $\mathcal Q^{++}$ remains the unique quadrant with the minimum $L_j$ value. Eventually, there exists a time $t'> t$ at which at least one parking node in $\mathcal Q^{++}$ becomes saturated. The rest of the proof follows from Case \ref{proof2-case1}. 
 \end{case}
 \begin{case} \label{proof2-case3}
 \normalfont
Condition $ \lnot$ $C_9 \land \lnot  C_{10} \land C_{11}$ holds. This implies that there exist at least two quadrants with the minimum value of $L_j$. However, the configuration is unbalanced. That is, there exists a unique quadrant with the minimum value of $L_j$ and with the maximum number of robot positions. $\mathcal Q^{++}$ is selected as the quadrant with the minimum $L_j$ value and the maximum number of robot positions. According to the execution of the algorithm, a robot from the quadrant $\mathcal Q^{++}$ will move towards one of the wedge boundaries only if each target parking node in $\mathcal Q^{++}$ is saturated. As a result, there exists an instant of time, where a candidate robot reaches a parking node, resulting in a unique quadrant with the minimum $L_j$ value. The rest of the proof follows from Case \ref{proof2-case2}.
 \end{case}
 \begin{case}
 \normalfont
 Condition $ \lnot$ $C_9 \land \lnot  C_{10} \land \lnot C_{11} \land C_{12}$ holds. The configuration is balanced. The robot with the highest order among all the robots lying on  $l \cup l'$ first moves along $l \cup l'$. Next, when it is one node away from the quadrant with the minimum $L_j$ value and the maximum number of robot positions, it moves towards an adjacent node. This results in transforming the configuration into an unbalanced configuration. The rest of the proof follows from Case \ref{proof2-case3}.  
 \end{case}
 \begin{case}
 \normalfont
 Condition $ \lnot$ $C_9 \land \lnot  C_{10} \land \lnot C_{11} \land \lnot C_{12} \land C_{13}$ holds. This implies the configuration is balanced and there do not exist any robot positions on $l \cup l'$. $\mathcal Q^{++}$ is selected as the quadrant that is not adjacent to the quadrant containing the guard, or the quadrant without the guard when exactly two quadrants contain the maximum number of robots and the minimum number of unsaturated parking nodes. The guard moves only when each parking node except for one becomes saturated. As a result, the quadrant non-adjacent to the quadrant containing the guard remains invariant. The target parking node will be selected in $\mathcal Q^{++}$ unless each parking node in $\mathcal Q^{++}$ becomes saturated. Eventually, there exists a time when each parking node in $\mathcal Q^{++}$ becomes saturated. Hence, $\mathcal Q^{++}$ remains invariant.
 \end{case}
 
\end{proof}
\begin{lemma} \label{chap6:lemma4}
If the configuration is such that the parking nodes admit a unique line of symmetry $l$, then during the execution of the algorithm Parking(), the target parking nodes remain invariant.
\end{lemma}
\begin{proof}
 The following cases are to be considered.
 \setcounter{case}{0}
 \begin{case}
 \normalfont
 There exists at least one unsaturated parking node on $l$. Let $p$ be the target parking node selected on $l$ in the TPS phase. As a result, $p$ is the unsaturated parking node on $l$, which has the highest order with respect to $\mathcal O_2$. Since the ordering $\mathcal O_2$ depends only on the positions of the fixed parking nodes, the ordering remains invariant while the robot moves towards it. Hence, $p$ remains the target parking node unless it becomes saturated. 
 \end{case}
 \begin{case}
     \normalfont
     Each parking node on $l$ is saturated and the configuration is asymmetric. The following subcases are to be considered.
     \setcounter{subcase}{0}
     \begin{subcase}
         \normalfont
         There exists an unsaturated parking node in $\mathcal H^{+}$. Let $p'$ be the target parking node in $\mathcal H^{+}$, selected in the TPS phase. According to Lemma \ref{chap6-lemma2}, $\mathcal H^{+}$ remains invariant unless each parking node in $\mathcal H^{+}$ becomes saturated. Since the ordering of the parking nodes is defined with respect to the leading corners and $\mathcal H^{+}$ remains invariant, $p'$ remains the target parking node unless it becomes saturated. Hence, the target parking node remains invariant. 
         \end{subcase}
         \begin{subcase}
           \normalfont
           Each parking node in $\mathcal H^{+}$ is saturated. Let $p''$ be the target parking node in $\mathcal H^{-}$, selected in the TPS phase. It should be noted that a parking node in $\mathcal H^{-}$ is selected as a target parking node only when each parking node in $\mathcal H^{+}$ becomes saturated. Since the ordering of the parking nodes is defined with respect to the leading corners and $\mathcal H^{+}$ remains invariant, $p''$ remains the parking node which has the highest order in $\mathcal H^{-}$ unless it becomes saturated. Hence, the target parking node remains invariant. 
         \end{subcase}
 \end{case}
 \begin{case}
     \normalfont
     Each parking node on $l$ is saturated, and the configuration is symmetric. There are two target parking nodes selected in two different half-planes. It should be noted that the ordering of the parking nodes is defined with respect to the leading corner, and the leading corners remain invariant while the robots move toward the target. While the candidate robots move towards the target parking nodes, they remain the unsaturated parking nodes with the highest order in their respective half-planes. Hence, the target parking node remains invariant. 
 \end{case}
 \end{proof}
 \begin{lemma} \label{chap6:lemma5}
   If the configuration is such that the parking nodes admit rotational symmetry, then during the execution of the algorithm Parking(), the target parking nodes remain invariant.  
 \end{lemma}
 \begin{proof}
     The following cases are to be considered.
     \setcounter{case}{0}
     \begin{case}
         \normalfont 
         There exists a parking node $p$ on $c$, which is unsaturated. Since $c$ is the center of rotational symmetry for the parking nodes, it remains invariant while the robot moves towards it. As a result, $p$ remains the target parking node unless it becomes saturated.
         \end{case}
     \begin{case}
         \normalfont
         There does not exist a parking node on $c$ or the parking node on $c$ is saturated. The configuration is asymmetric. First, assume that $p \neq p_{guard}$ is the target parking node. The following subcases are to be considered.
         \setcounter{subcase}{0}
         \begin{subcase}
             \normalfont
             The target parking node is selected in $\mathcal Q^{++}$. Let $p'$ be the target parking node selected in $\mathcal Q^{++}$. As a result, $p'$ is the unsaturated parking node in $\mathcal Q^{++}$, which has the highest order among all the unsaturated parking nodes in $\mathcal Q^{++}$. According to Lemma \ref{chap6-lemma3}, $\mathcal Q^{++}$ remains invariant unless each parking node in $\mathcal Q^{++}$ becomes saturated. It should be noted that the leading corners are defined with respect to the position of fixed parking nodes. The leading corners remain invariant while the robot moves toward the target parking nodes. Since the ordering of the parking nodes is defined with respect to the leading corners and $\mathcal Q^{++}$ remains invariant, $p'$ remains the target parking node unless it becomes saturated. Hence, the target parking node remains invariant.
         \end{subcase}
         \begin{subcase}
             \normalfont
             The target parking node is selected on the wedge boundaries corresponding to $\mathcal Q^{++}$. This implies that each parking node in $\mathcal Q^{++}$ becomes saturated. As $\mathcal Q^{++}$ remains invariant according to Lemma \ref{chap6-lemma3}, the wedge boundaries corresponding to $\mathcal Q^{++}$ remain invariant. Since the leading corners are defined with respect to the position of fixed parking nodes, the unsaturated parking node with the highest order on the wedge boundaries remains invariant. Hence, the target parking node remains invariant.
         \end{subcase}
         \begin{subcase}
             \normalfont
             The target parking node is selected on a quadrant adjacent to $\mathcal Q^{++}$. This implies that each parking node in the wedge boundaries corresponding to $\mathcal Q^{++}$ becomes saturated. It should be noted that since $\mathcal Q^{++}$ remains invariant, the quadrants adjacent to $\mathcal Q^{++}$ remain invariant. As the leading corners are defined with respect to the position of fixed parking nodes, the unsaturated parking node with the highest order on the quadrants $\mathcal Q^{+-}$ and $\mathcal Q^{-+}$ remains invariant. Hence, the target parking node remains invariant.
         \end{subcase}
         \begin{subcase}
             \normalfont
             The target parking node is selected on the wedge boundaries corresponding to $\mathcal Q^{+-}$ and $\mathcal Q^{-+}$. Since $\mathcal Q^{+-}$ and $\mathcal Q^{-+}$ remain invariant, the wedge boundaries corresponding to $\mathcal Q^{+-}$ and $\mathcal Q^{-+}$ remain invariant. As a result, the unsaturated parking node having the highest order on the wedges remains invariant. Hence, the target parking node remains invariant.  
         \end{subcase}
         \begin{subcase}
         \normalfont
             The target parking node is selected in $\mathcal Q^{--}$. This implies that each parking node on the other quadrants and on the wedge boundaries becomes saturated. As a result, at this stage of the algorithm, $\mathcal Q^{--}$ is the unique quadrant that has unsaturated parking nodes. The target parking node is selected in $\mathcal Q^{--}$, which is unsaturated and has the highest order in $\mathcal Q^{--}$. Since $\mathcal Q^{--}$ is the quadrant non-adjacent to $\mathcal Q^{++}$ and $\mathcal Q^{++}$ remains invariant, the target parking node also remains invariant unless it becomes saturated. Hence, the target parking node remains invariant.  
         \end{subcase}
         \noindent If $p_{guard}$ is the target parking node, it would be selected as the target parking node only when each parking node becomes saturated. As a result, it remains invariant. 
     \end{case}
     \begin{case}
       \normalfont
         There does not exist a parking node on $c$ or the parking node on $c$ is saturated. The configuration is symmetric with respect to rotational symmetry. In that case, the target parking nodes are selected as the unsaturated parking nodes with the highest order. It should be noted that these targets are selected either on the quadrants or on the wedge boundaries, and at any given time, there is a maximum of four parking nodes selected as targets. Since the leading corners are defined with respect to the position of fixed parking nodes, the parking nodes having the highest orders remain invariant. Hence, the target parking nodes remain invariant. 
         \end{case}
        \end{proof}
  \begin{lemma} \label{chap6:lemma6}
     During the CRS phase, the candidate robot remains invariant.
 \end{lemma}
 \begin{proof}
     The following cases are to be considered.
     \setcounter{case}{0}
     \begin{case}
         \normalfont
         $C(t)$ is asymmetric. The candidate robot is selected as the robot that has a shortest free path toward the target parking node. If there exists a unique such robot, while the candidate robot moves toward its target, it remains the unique robot that has a shortest free path towards its target. Hence, it remains invariant. Next, consider the case when there are multiple such robots. Since $C(t)$ is asymmetric, the robots are orderable. The candidate is selected as the robot which has the highest order among such robots. While the candidate moves an adjacent node toward the target, it becomes the unique robot that has a shortest free path toward the target. The rest of the proof follows from the previous case when there exists a unique robot that has a shortest free path toward the target parking node. Hence, the candidate robot remains invariant.
     \end{case}
     \begin{case}
         \normalfont
         $C(t)$ is symmetric with respect to a single line of symmetry $l$ and the target parking node is on $l$. The candidate robots are selected as the two symmetric robots that have a shortest free path toward $p$. While they move toward $p$, they remain the robots with a shortest free path toward $p$. If there are multiple such robots, the candidate robots are the two robots with the highest order in their respective half-planes and have a shortest free path toward $p$. While they move towards $p$, they remain the unique robots in their respective half-planes that have a shortest free path towards $p$. Hence, the candidate robot remains invariant. 
     \end{case}
     \begin{case}
         \normalfont
         $C(t)$ is symmetric with respect to a single line of symmetry $l$ and the target parking nodes are on the half-planes. It should be noted that in this case, a candidate robot is selected in each half-plane, for each target parking node. The candidates are selected as the robots which have a shortest free path toward the target. While the candidates move toward their respective targets, it remains the unique robot that has a shortest free path toward their target. The proof proceeds similarly to the previous case when there are multiple such robots for each target parking node. Hence, the candidate robot remains invariant. 
     \end{case}
     \begin{case}
         \normalfont
         $C(t)$ is symmetric with respect to rotational symmetry and the target parking node is on $c$. Depending, on whether the angle of rotation is $90 ^{\circ}$ or $180 ^{\circ}$, two or four candidate robots are selected. The robots with the shortest free path toward the target are selected as candidate robots. While they move toward the target, they remain the robots that have a shortest free path toward the target. Hence, the candidate robot remains invariant. 
     \end{case}
     \begin{case}
         \normalfont
         $C(t)$ is symmetric with respect to rotational symmetry and the target parking nodes are on the quadrants or on the wedge boundaries. It should be noted that, in this case, a candidate robot is selected in each quadrant or on the wedge boundaries, for each target parking node. The candidates are selected as the robots which have a shortest free path toward the target. While the candidates move toward their respective targets, it remains the unique robot that has a shortest free path toward their target. Hence, the candidate robot remains invariant. 
     \end{case}
    \end{proof}
   \noindent  We next consider the following theorem.
\begin{theorem}
    Algorithm Parking() solves the Parking Problem in Infinite grids for all initial configurations not belonging to the set $\mathcal U$.
\end{theorem}
\begin{proof}
\normalfont
The algorithm \textit{Parking()} proceeds according to the different classes of configurations. The following cases are to be considered.
\setcounter{case}{0}
\begin{case}
    \normalfont
    $ C(0) \in \mathcal I_1$: Since the parking nodes are asymmetric, the parking nodes are uniquely orderable. The ordering remains invariant while the robots move toward the target parking node. The target parking node can be selected in a unique manner. According to Lemma \ref{chap6:lemma6}, the candidate robot remains invariant while it moves toward the target parking node. While the target parking node contains a number of robots equal to the capacity of the parking node, it becomes saturated. The process continues unless each parking node becomes saturated. Hence, the algorithm Parking() solves the parking problem for all configurations belonging to $\mathcal I_1$.
\end{case}
\begin{case}
    \normalfont 
    $C(0) \in \mathcal I_2$: First consider the case when $C(0) \in \mathcal I_{21}$. In the GS phase, a guard is selected and placed in such a way that the configuration remains asymmetric during the execution of the algorithm. As a result, the configuration $C(t) \notin \mathcal U$, for any $t>0$. The target parking node is selected in the TPS phase, which remains invariant according to Lemma \ref{chap6:lemma4}. The candidate robot moves toward the target parking node. According to Lemma \ref{chap6:lemma6}, the candidate robot remains invariant while it moves toward the target parking node. When all the parking nodes become saturated, the algorithm terminates. Next, consider the case when $\mathcal C(0) \in \mathcal I_{221} $. The symmetry of the configuration can be broken by allowing the robot on $l$ with the maximum configuration view is allowed to move towards an adjacent node away from $l$. This transforms the configuration into an asymmetric configuration and the procedure proceeds similarly as before. If $C(0) \in \mathcal I_{222}$, the target parking node is selected in the TPS phase, which remains invariant according to Lemma \ref{chap6:lemma4}. The candidate robot moves toward the target parking node. According to Lemma \ref{chap6:lemma6}, the candidate robot remains invariant while it moves toward the target parking node. The procedure terminates when all the parking nodes become saturated. In case $C(0) \in \mathcal I_{223}$, first the parking nodes on $l$ become saturated and the rest of the procedure proceeds similarly as in the case when $C(0) \in \mathcal I_{222}$.
    
\end{case}
\begin{case}
    \normalfont
    $\mathcal C(0) \in \mathcal I_3$. First, consider the case when $C(0) \in \mathcal I_{31}$. The placement of the guard in the Guard Selection and Placement phase ensures that the configuration remains asymmetric during the execution of the algorithm Parking(). As a result, the configuration $C(t) \notin \mathcal U$, for any $t>0$. The target parking node is selected in the TPS phase, which remains invariant according to Lemma \ref{chap6:lemma5}. The candidate robot moves toward the target parking node. According to Lemma \ref{chap6:lemma6}, the candidate robot remains invariant while it moves toward the target parking node. When all the parking nodes become saturated, the algorithm terminates. If $C(0) \in \mathcal I_{321}$, the symmetry of the configuration can be broken by allowing the robot on $c$ to move towards an adjacent node. This transforms the configuration into an asymmetric configuration. The procedure Parking() proceeds similarly as before. If $C(0) \in \mathcal I_{322}$, the target parking node is selected in the TPS phase, which remains invariant according to Lemma \ref{chap6:lemma5}. The candidate robot moves toward the target parking node. According to Lemma \ref{chap6:lemma6}, the candidate robot remains invariant while it moves toward the target parking node. The procedure terminates when all the parking nodes become saturated. In case $C(0) \in \mathcal I_{323}$, first the parking nodes on $c$ become saturated and the rest of the procedure proceeds similarly as in the case when $C(0) \in \mathcal I_{322}$.
\end{case}
\end{proof}
\section{Conclusion} \label{chap6:sec7}
\noindent This chapter proposed a deterministic distributed algorithm for solving the parking problem in infinite grids. We have characterized all the initial configurations and the values of $k_i$ for which the problem is unsolvable, even if the robots are endowed with strong multiplicity detection capability. A deterministic algorithm has been proposed under the assumption that the robots are endowed with global-strong multiplicity detection capability. As a future work, it would be interesting to investigate the problem in case the number of robots is not equal to the sum of the capacities of the parking nodes. In case the number of robots in the initial configuration is less than the sum of the capacities of the parking nodes, one interesting study could be to investigate the problem with the objective of maximizing the number of saturated parking nodes. Another direction of future work would be to consider the problem with the objective of minimizing the number of moves in order to accomplish the parking process.


 \bibliographystyle{elsarticle-num} 
 \bibliography{ref.bib}


\end{document}